\documentclass[12pt]{amsart}

\usepackage{epsfig}
\usepackage{verbatim}
\usepackage{amssymb}
\usepackage{latexsym}
\usepackage{amsmath}
\usepackage{amsfonts}
\usepackage{amsbsy}
\usepackage{amscd}
\usepackage{graphicx}
\usepackage[usenames]{color}
\usepackage[hypertex]{hyperref}%
\usepackage[margin=1in]{geometry}

\newcommand{\PP}{{\mathcal P}}

\newcommand{\E}{{\mathcal E}}

\newcommand{\SSS}{{\mathcal S}}

\newcommand{\I}{{\mathcal I}}
\newcommand{\J}{{\mathcal J}}

\newcommand{\R}{{\mathbb R}}

\newcommand{\N}{{\mathbb N}}

\renewcommand{\eqref}[1]{(\ref{#1})}

\newcommand{\supp}{{\rm supp}}

\newcommand{\Dom}{{\rm Dom}\,}

\newtheorem{prop}{Proposition}[section]
\newtheorem{lemma}[prop]{Lemma}

\newtheorem{defi}{Definition}[section]
\newtheorem{coro}[prop]{Corollary}
\newtheorem{theo}[prop]{Theorem}

\newtheorem{rmk}[prop]{Remark}

\begin{document}

\setcounter{figure}{0}

\title[One-dimensional wave equations defined by fractal Laplacians]
{One-dimensional wave equations defined by fractal Laplacians}

\date{\today}

\author[J. F.-C. Chan]{John Fun-Choi Chan}
\address{Department of Mathematical Sciences\\ Georgia Southern
University\\ Statesboro, GA 30460-8093, USA.} 
\email{fchan@udel.edu}

\author[S.-M. Ngai]{Sze-Man Ngai}
\address{Department of Mathematical Sciences\\ Georgia Southern
University\\ Statesboro, GA 30460-8093, USA, and College of Mathematics and Computer Science, Hunan Normal University,
Changsha, Hunan 410081, China.}
\email{smngai@georgiasouthern.edu}

\author[A. Teplyaev]{Alexander Teplyaev}
\address{Department of Mathematics\\
University of Connecticut\\
Storrs, CT 06269, USA.}
\email{alexander.teplyaev@uconn.edu}

\subjclass[2010]{Primary: 28A80, 34L16, 65L60, 65L20; Secondary: 34L15, 65L15}
\keywords{fractal, iterated function system, self-similar measure, Laplacian, wave equation, finite element method}

\thanks{The first two authors were supported in part by a Faculty Research Grant from Georgia Southern University. The second author is also supported in part by an HKRGC grant and
National Natural Science Foundation of China grant 11271122.
The third author is supported in part by  NSF grant DMS-0505622.}

\begin{abstract}
We study one-dimensional wave equations defined by a class of fractal Laplacians.
These Laplacians are defined by fractal measures generated by iterated function systems with overlaps, such as the well-known infinite Bernoulli convolution associated with the golden ratio and the 3-fold convolution of the Cantor measure. The iterated function systems defining these measures do not satisfy the  post-critically finite condition or the open set condition. By using second-order self-similar identities introduced by Strichartz et al., we discretize the equations and use the finite element and central difference methods to obtain numerical approximations to the weak solutions. We prove that the numerical solutions converge to the weak solution, and obtain estimates for the rate of convergence. \tableofcontents
\end{abstract}

\maketitle

\section{Introduction}\label{S:intro}
\setcounter{equation}{0}

In this paper we study approximations to the solution of the wave equation defined by a one-dimensional fractal measure. Such fractals have recently attracted considerable attention because of their relation to classical analysis on one hand, and having many unusual properties on the other hand. In such situations classical approximation methods have to be modified to produce accurate results, see \cite[and references therein]
{St2004num,MV,St2000sp, Meyers-Strichartz-Teplyaev_2004}. In this paper we investigate the solution of the wave equation theoretically, and also provide numerical examples.

Our long term goal is to combine ideas of Strichartz,
including the celebrated Strichartz estimates, with some recent results,
such as \cite{HSTZ} and \cite{Strichartz-Teplyaev_2012},
in a comprehensive study of wave equations on fractals and fractafolds.
 However currently there are few  mathematical tools developed to study wave equations on fractals, despite the fact that
the existence of large gaps in the spectrum on many fractals, together with heat kernel estimates, implies that Fourier series on these fractals can have better convergence than in the classical case (and, as was noted by Strichartz in \cite{St}, \textquotedblleft ... is the first kind of example which improves on the corresponding results in smooth analysis\textquotedblright).
Among the most recent results, the infinite wave prorogation speed was recently proved
on some \textit{post-critically finite (p.c.f.)} (see \cite{Kigami_2001})  fractals in the preprint \cite{Lee} by Yin-Tat Lee. This question was open, even in the most standard case of the Sierpi\'nski gasket,
since 1999, see \cite{DSV, Strichartz_1999}.
The proof in \cite{Lee} relies partially on
the Kigami's p.c.f. assumptions (see \cite{Kigami_2001} and references therein),
and more substantially on certain heat kernel estimates.
In general, the heat kernel estimates on fractals is a difficult and extensively studied subject,
with most relevant recent results and references contained in   \cite{GT12,GT-cjm,Ki09,Ki12}.
It is not clear at present if the heat kernel
estimates assumed in \cite{Lee}  can be verified for fractal measures that we consider,
but some preliminary results can be found in \cite{PLYung}.
The most intuitive idea, essentially due to Strichartz, is that there is no reason why the wave propagation speed should be finite on fractals, because of the difference in time and Laplacian scalings.
In our paper we do not discuss the wave propagation speed directly, but rather develop approximating tools that may help in this study.

Let $\mu$ be a continuous positive finite Borel measure on $\R$ with $\text{supp}(\mu)\subseteq [a,b]$. Let $L^2_{\mu}[a,b]:=\{f:[a,b]\to\R:\int|f|^2\,d\mu<\infty\}$; if $\mu$ is Lebesgue measure, we simply denote the space by $L^2[a,b]$. It is well known (see e.g., \cite{Bird-Ngai-Teplyaev_2003, Hu-Lau-Ngai_2006}) that $\mu$ defines a Dirichlet Laplacian $\Delta_\mu$ on $L^2_{\mu}[a,b]=\{f:[a,b]\to\R:\int|f|^2\,d\mu<\infty\}$, described as follows. Let $H^1(a,b)$ be the Sobolev space of all functions in $L^2[a,b]$ whose weak derivatives  belong to $L^2[a,b]$, with the inner product
$$
(u,v)_{H^1(a,b)}:=\int_a^buv\,dx+\int_a^b u'v'\,dx.
$$
Let $H_0^1(a,b)$ be the completion of $C_c^\infty(a,b)$ in $H^1(a,b)$. $H_0^1(a,b)$ and $H^1(a,b)$ are dense subspaces of $L^2_{\mu}[a,b]$.
Define a quadratic form on $L^2_{\mu}[a,b]$,
\begin{equation}\label{eq:quadratic-form}
\E(u,v)=\int_a^bu'v'\,dx,
\end{equation}
with domain $\Dom\E$ equal to some dense subspace of $H_0^1(a,b)$ (see Section~\ref{S:prelim}). Since the embedding $\Dom\E\hookrightarrow L^2_{\mu}[a,b]$ is compact, $\E$ is closed and is in fact a Dirichlet form on $L^2_{\mu}[a,b]$. Thus there exists a nonnegative self-adjoint operator $T$ on $L^2_{\mu}[a,b]$ such that $\Dom\E=\Dom(T^{1/2})$ and
$$
\E(u,v)=(T^{1/2}u,T^{1/2}v)_{\mu}\qquad \text{for all } u,v\in\Dom\E,
$$
where
$$
(u,v)_\mu:=\int_a^b uv\,d\mu
$$
denotes the inner product on $L^2_{\mu}[a,b]$; we will also let $\|\cdot\|_\mu$ denote the corresponding norm. We define $\Delta_\mu:=-T$ and call it the {\it Dirichlet Laplacian with respect to $\mu$}.

Let $u\in\Dom\E$ and $f\in L^2_{\mu}[a,b]$. It is known that $u\in\Dom(\Delta_\mu)$ and $\Delta_\mu u=f$ if and only if $\Delta u=fd\mu$ in the sense of distribution, i.e.,
$$
\int_a^bu'v'\,dx=\int_a^b(-\Delta_\mu u)v\,d\mu\quad\text{for all }v\in C_c^\infty(a,b).
$$
It is also known (see, e.g., \cite{Bird-Ngai-Teplyaev_2003, Hu-Lau-Ngai_2006}) that there exists an orthonormal basis of eigenfunctions of $\Delta_\mu$ and the eigenvalues $\{\lambda_n\}$ are discrete and satisfy $0\le\lambda_1<\lambda_2<\cdots$ with  $\lim_{n\to\infty}\lambda_n=\infty$.

The operators $\Delta_\mu$ and their generalizations have been studied in connection with spectral functions of the string and diffusion processes (see \cite{Feller_1955,Feller_1957,Kac-Krein_1974}). More recently, they have been studied in connection with fractal measures (see \cite{Freiberg_2003, Freiberg_2004_2005,  Freiberg_2005, Freiberg-Lobus_2004, Fujita_1987,  Bird-Ngai-Teplyaev_2003,Hu-Lau-Ngai_2006, Naimark-Solomyak_1994, Naimark-Solomyak_1995,  Solomyak-Verbitsky_1995, Ngai_2011}).

Our study of the operator $\Delta_\mu$ is mainly motivated by the effort to extend the current theory of analysis on fractals to include iterated function systems (IFSs) with overlaps. Such IFSs do not satisfy the well-known post-critically finite condition or the open set condition. Nevertheless, by assuming $\mu$ satisfies a family of second-order self-similar identities, some results concerning $\Delta_\mu$ can be obtained. In \cite{Chen-Ngai_2010}, the finite element method is used to compute numerical approximations to the eigenvalues and eigenfunctions, and in \cite{Ngai_2011}, formulas defining the spectral dimension of $\Delta_\mu$ have been obtained for a class of measures that include the infinite Bernoulli convolution associated with the golden ratio and the three-fold convolution of the Cantor measure.

The main purpose of this paper is to study one-dimensional wave equations defined by a class of fractal Laplacians, subject to the Dirichlet boundary condition. More precisely, we study the following non-homogeneous \textit{hyperbolic initial/boundary value problem (IBVP)}:
\begin{equation}\label{eq:IBVP}
\left \{
\begin{aligned}
        u_{tt} &- \Delta_{\mu}u = f  \quad  &&\text{on } [a,b] \times [0,T],\\
        u&=0 \quad &&\text{on }\{a, b\} \times [0,T] ,\\
        u&=g,\  u_t=h \quad &&\text{on } [a,b] \times \{t=0\}.
\end{aligned}
\right.
\end{equation}

The following existence and uniqueness result (see Definition~\ref{D:weak-solution} for the definition of a \textit{weak solution}) follows easily from the general theory for wave equations in Hilbert spaces (see Section~\ref{S:prelim}).

\begin{theo}\label{T:main-1}
Assume $g\in \Dom\E, h \in L^2_{\mu}[a,b] \text{ and } f \in L^2(0,T;\Dom\E)$. Then equation (\ref{eq:IBVP}) has a unique weak solution.
\end{theo}

We are mainly interested in fractal measures $\mu$. Let $D$ be a non-empty compact subset of $\R^d$. A function $S : D \rightarrow D$ is called a \emph{contraction} on $D$ if there is a number $c$ with $0<c<1$ such that
\begin{equation}\label{eq:contraction}
|S(x)- S(y)|\leq c|x-y|\quad\text{for all } x,y\in D.
\end{equation}
An \emph{iterated function system (IFS)} on $D$ is a finite collection of contractions on $D$. Each IFS $\{S_i\}_{i=1}^q$ defines a unique compact subset $F\subseteq D$, called the \textit{invariant set} or \textit{attractor}, satisfying
$$
F=\bigcup_{i=1}^q S_i(F).
$$
Also, to each set of probability weights $\{p_i\}_{i=1}^q$, where $p_i>0$ and $\sum_{i=1}^q p_i =1$, there exists a unique probability measure, called the \emph{invariant measure}, satisfying the identity
\begin{equation}\label{self-similarmeasure}
\mu = \sum_{i=1}^q p_i\mu \circ S^{-1}_i
\end{equation}
(see \cite{Hutchinson_1981, Falconer_1990}).
$S$ is a contractive \textit{similitude} if equality in \eqref{eq:contraction} holds. IFSs studied in this paper consist of contractive similitudes; they are of the form
\begin{equation}\label{eq:IFS_similitudes}
S_i(x)=\rho_i R_i x +b_i,\qquad i=1,\dots,q,
\end{equation}
where $0<\rho_i<1$, $R_i$ is an orthogonal transformation, and $b_i\in\R^d$. For such an IFS, we call the corresponding invariant set $F$ the \textit{self-similar set} and the invariant measure $\mu$ the \emph{self-similar measure}.

An IFS $\{S_i\}_{i=1}^q$ is said to satisfy the \textit{open set condition (OSC)} if there exists a non-empty bounded open set $U$ such that $\cup_i S_i(U)\subseteq U$ and $S_i(U)\cap S_j(U)=\emptyset$ if $i\ne j$. An IFS that does not satisfy the OSC is said to have \textit{overlaps}. For an IFS of contractive similitudes, it is known that if the linear parts of the IFS maps are commensurable, then the p.c.f. condition implies the OSC~\cite{Deng-Lau_2008}.

We are interested in one-dimensional self-similar measures defined by IFSs with overlaps. Such IFSs are not p.c.f. and are thus not covered by Kigami's theory. In order to discretize a wave equation and obtain numerical approximations to the weak solution, we will assume that the corresponding self-similar measure satisfies a family of second-order self-similar identities, an idea introduced by Strichartz \textit{et al.} \cite{Strichartz-Taylor-Zhang_1995}. Let $\{S_i\}_{i=1}^q$ be an IFS of contractive similitudes on $\R$ and let $\mu$ be the corresponding self-similar measure. Assume, in addition, that $\supp(\mu)=[a,b]$. Define an auxiliary IFS
\begin{equation}\label{eq:IFS_T_j}
    T_j(x)=\rho^{n_j}x +d_j,\quad j=1,2,\ldots,N,
\end{equation}
where $n_j\in\mathbb N$ and $d_j\in\mathbb R$, and let
\begin{equation}\label{eq:rho_def}
\rho:=\max\{\rho^{n_j}:1\le j\le N\}.
\end{equation}
$\mu$ is said to
satisfy a family of {\it second-order self-similar identities} (or
simply {\it second-order identities}) with respect to
$\{T_j\}_{j=1}^{N}$ (see \cite{Lau-Ngai_2000}) if
\begin{enumerate}
\item[(i)]
$\text{supp}(\mu)\subseteq \bigcup_{j=1}^{N}
T_j(\text{supp}(\mu))$, and
\item[(ii)] for each Borel subset
$A\subseteq\text{supp}(\mu)$ and $0\le i,j\le N$,
$\mu(T_iT_jA)$ can be expressed as a linear combination of
$\{\mu(T_kA): k=1,\ldots,N\}$ as
\begin{equation*}
    \mu(T_iT_jA)=\sum_{k=0}^N d_k\mu (T_kA),
\end{equation*}
where $d_k = d_k(i,j)$ are independent of $A$. In matrix form,
\begin{equation}\label{2ndorderidentity}
\begin{aligned}
    &\left[ \begin{array}{c}
            \mu(T_1T_jA)\\
            \vdots\\
            \mu(T_{N}T_jA)
    \end{array}
    \right]
    =M_j\left[\begin{array}{c}
            \mu(T_1A)\\
            \vdots\\
            \mu(T_{N}A)
    \end{array}
    \right], \quad j=1,\ldots,N, \qquad\text{or equivalently,}\\
    &\mu(T_iT_jA)=\emph{\textbf{e}}_i M_j
                    \left[\begin{array}{c}
            \mu(T_1A)\\
            \vdots\\
            \mu(T_{N}A)
            \end{array}\right],\quad i,j=1,\dots,N,
\end{aligned}
\end{equation}
where $\emph{\textbf{e}}_i$ is the $i$th row of the $N\times N$
identity matrix, and $M_j$ is some $N\times N$ matrix independent of $A$.
\end{enumerate}

For our purposes, we will assume that $\{T_j\}_{j=0}^N$ satisfies the OSC. The $m$-th level iteration of the auxiliary IFS $\{T_j\}_{j=0}^N$ induces a partition $V_m$ of  ${\rm supp}(\mu)=[a,b]$. Moreover, the $\mu$ measure of each subinterval in the partition can be computed in terms of a matrix product. This provides us with a good way to discretize the wave equation.

By letting $f(x,t)=0$, multiplying the first equation in (\ref{eq:IBVP}) by  $v\in \Dom\E$,  integrating both sides with respect to $d\mu$, and then using integration by parts, we obtain
\begin{equation}\label{eq:PDE-weakform-main}
 -\int^{b}_{a} u_x(x,t)\,v'(x) \, dx = \int^{b}_{a}
u_{tt}(x,t)\,v(x) \,d\mu,
\end{equation}
where $u_x(x,t)$ is the weak partial derivative of $u$ with respect to $x$ and $u_{tt}$ is the weak second partial derivative with respect to $t$.

\begin{theo}\label{T:main-2}
Let $\mu$ be a self-similar measure defined by a one-dimensional IFS of contractive similitudes on $\R$ as in (\ref{self-similarmeasure}) and  \eqref{eq:IFS_similitudes}. Assume that {\rm supp}$(\mu)=[a,b]$ and that $\mu$
satisfies a family of second-order self-similar
identities.
Then the finite element method for the equation (\ref{eq:PDE-weakform-main})
discretizes it
 to a system of second-order ordinary differential equations \eqref{eq:IVP}, which
has a unique solution (and can be solved numerically).
\end{theo}

Based on this result, we solve the homogeneous IBVP \eqref{eq:IBVP} numerically for three different measures, namely, the weighted Bernoulli-type 
measure, the infinite Bernoulli convolution associated with the golden ratio, and the 3-fold convolution of the Cantor measure. We show that the approximate solutions converge to the actual weak solution and obtain a rate of convergence.

\begin{theo}\label{T:main-3}
Assume the same hypotheses as in Theorem~\ref{T:main-2}. Let $f=0$ in equation \eqref{eq:IBVP} and fix $t\in [0,T]$. Then the approximate solutions $u^m$ obtained by the finite element method converge in $L^2_\mu[a,b]$ to the actual weak solution $u$. Moreover, there exists a constant $C>0$ such that for all $m\ge 1$,
$$
\left\|u^m-u\right\|_{\mu}\leq \big(C\sqrt{T}\left\|u_{tt}\right\|_{2,\Dom\E}+2\left\|u\right\|_{\Dom\E}\big)\rho^{m/2}.
$$
\end{theo}

This paper is organized as follows.
We summarize some basic classical results, definitions, and notation in Section~\ref{S:prelim}.
In Section~\ref{S:FEM} we use the finite element and central difference methods to obtain numerical
approximations to the corresponding homogeneous IBVP (\ref{eq:IBVP}), proving Theorem~\ref{T:main-2}.
In Section~\ref{S:Frac} we apply our numerical methods to the above-mentioned measures, and
illustrate some numerical results.
In Section~\ref{S:convergence} we prove the convergence of the approximation scheme and obtain an estimate for the convergence rate stated in Theorem~\ref{T:main-3}.

\section{Preliminaries}\label{S:prelim}
\setcounter{equation}{0}

In this section, we summarize some notation, definitions, and preliminary results that will be used throughout the rest of the paper. For a Banach space $X$, we denote its topological dual by $X'$.
For  $v\in X'$ and $u\in X$ we let $\langle v,u\rangle=\langle v,u\rangle_{X',X}:=v(u)$ denote the \textit{dual pairing} of $X'$ and $X$.

A function $s : [0,T] \rightarrow X$ is called \textit{simple} if it has the form
\begin{equation}
s(t) = \sum_{m=1}^N \chi_{E_m}(t) u_m\quad  \text{for }t \in[0, T],
\end{equation}
where each ${E_m}$ is a Lebesgue measurable subset of $[0,T]$, $u_m\in X$ for $m=1,\dots,N$, and $\chi_{E_m}$ is the characteristic function on $E_m$. A function $u : [0,T] \rightarrow X$ is \textit{strongly measurable} if there exist simple functions $s_n:[0,T] \rightarrow X$ such that
$$
s_n(t) \rightarrow u(t)\quad \text{ as $n\to\infty$ for Lebesgue a.e. } t\in [0,T].
$$
A function $u : [0,T] \rightarrow X$ is \textit{weakly measurable} if for each $v\in X'$, the mapping
$t\mapsto \left\langle v, u(t) \right\rangle $ is Lebesgue measurable.

A function $ u: [0,T] \rightarrow X$ is \textit{almost separably valued} if there exists a subset $E \subseteq [0,T]$ with zero Lebesgue measure such that the set $ \{u(t): t\in [0,T] \backslash E \}$ is separable. By a theorem of Pettis \cite{Pettis_1938}, a function $u: [0,T] \rightarrow X$ is strongly measurable if and only if it is weakly measurable and almost separably valued. Since any subset of a separable Banach space $X$ is separable, the two concepts of measurability coincide and we can use the term \textit{measurable} without ambiguity.

\begin{defi}\label{D:banach-space}
Let $X$ be a separable Banach space with norm $\left\|\cdot \right\|_X. $
Define $L^p(0,T;X)$ to be the space of all measurable functions $u: [0,T] \rightarrow X$ satisfying
\begin{enumerate}
\item[(a)] $\displaystyle\left\| u \right\|_{ L^p(0,T;X)} := \Big( \int_0^{T} \left\|u(t)\right\|_X^p\,dt\Big)^\frac{1}{p}<\infty$,
if $1\leq p < \infty$, and

\item[(b)] $\displaystyle\left\| u \right\|_{ L^\infty(0,T;X)} :=  {\rm ess{\,} sup}_{0\leq t \leq T}  \left\| u(t) \right\|_X<\infty$, if $p=\infty$.
\end{enumerate}
If the interval $[0,T]$ is understood, we will abbreviate these norms as $\left\| u \right\|_{p,X}$ and
$\left\| u \right\|_{\infty,X}$.
\end{defi}

\begin{rmk}\label{rmk:Banach-Hilbert} For each $1\leq p \leq \infty$, $L^p(0,T;X)$ is a Banach space; moreover, $L^{p_2}(0,T;X)\subseteq L^{p_1}(0,T;X)$ if $0\le p_1\le p_2\le\infty$. If $(X,(\cdot,\cdot)_X)$ is a separable Hilbert space, then $L^2(0,T;X)$ is a Hilbert space with the inner product
$$
(u,v)_{L^2(0,T;X)}:=\int_0^T \big(u(t),v(t)\big)_X\,dt.
$$
\end{rmk}

\begin{defi}\label{D:weak-derivative}
Let $X$ be a Banach space and $u \in L^1(0,T;X)$. We say $v \in L^1(0,T;X)$ is the {\emph{weak derivative}} of $u$, written $u_t=v$, if
$$
\int_0^T \phi_t(t) u(t) \,dt =- \int_0^T \phi(t) u_t(t)\,dt
$$
for all scalar test functions $\phi \in C_c^{\infty}(0,T)$.
\end{defi}

\begin{defi}\label{D:weak-convergence}
Let $X$ be a Banach space and $X'$ its dual. We say a sequence $\{u_m\}_{m=1}^{\infty}\subseteq X$ {\emph{converges weakly to}} $u \in X$, written $ u_m \rightharpoonup u$,
if $$\left\langle v, u_m\right\rangle \rightarrow \left\langle v, u \right\rangle $$
 for each bounded linear functional $v \in X'$.
\end{defi}
For the more general definition of derivatives of distributions with values in a Hilbert space, we refer the reader to \cite[Section 25]{Wloka_1987}.

The notion of a Gelfand triple \cite{Gelfand-Vilenkin_1964}, defined below, plays an important role in our investigation of the wave equation.

\begin{defi}\label{D:Gelfand-triple}
Let $V$ and $H$ be separable Hilbert spaces with the continuous injective dense embedding $\iota:V\hookrightarrow H$. By identifying $H$ with its dual $H'$, we obtain the following continuous and dense embedding:
$$
V\hookrightarrow H\cong H'\hookrightarrow V'.
$$
Assume in addition that the dual pairing between $V$ and $V'$ is compatible with the inner product on $H$, in the sense that
$$
\langle v,u\rangle_{V',V}=(v,u)_H
$$
for all $u\in V\subset H$ and $v\in H\cong H'\subset V'$. The triple $(V,H,V')$ is called a {\rm Gelfand triple} (The pair $(H,V)$ is also called a {\rm rigged Hilbert space}.)
\end{defi}

We remark that since $V$ is itself a Hilbert space, it is isomorphic with its dual $V'$. However, this isomorphism is in general not the same as the composition $\iota^*\iota:V\subset H=H'\hookrightarrow V'$, where $\iota^*$ is the adjoint of $\iota$.

Throughout the rest this section we let $\mu$ be a finite positive Borel measure on $\R$ with ${\rm supp}(\mu)\subseteq[a,b]$ and $\mu(a,b)>0$, where $-\infty<a<b<\infty$. It is known that the following important condition is satisfied (see e.g., \cite{Hu-Lau-Ngai_2006, Naimark-Solomyak_1994}):
There exists a constant $C>0$ such that
\begin{equation}\label{eq:condition_C}
\int_a^b|u|^2\,d\mu\le C\int_a^b |\nabla u|^2\,dx\quad\text{for all }
u\in C_c^\infty(a,b).
\end{equation}
This condition implies that each equivalence class $u\in
H_0^1(a,b)$ contains a unique (in $L^2_{\mu}[a,b]$ sense) member
$\bar u$ that belongs to $L^2_{\mu}[a,b]$ and satisfies both
conditions below:
\begin{enumerate}
\item[(1)] There exists a sequence $\{u_n\}$ in $C_c^\infty(a,b)$ such that
$u_n\to\bar u$ in $H_0^1(a,b)$ and $u_n\to\bar u$ in $L^2_{\mu}[a,b]$;

\item[(2)] $\bar{u}$ satisfies the inequality in \eqref{eq:condition_C}.
\end{enumerate}
We call
$\bar u$ the {\it $L^2_{\mu}[a,b]$-representative} of $u$. Assume condition \eqref{eq:condition_C} holds
and define a mapping
$\iota:H_0^1(a,b)\to L^2_{\mu}[a,b]$ by
$$
\iota(u)=\bar u.
$$
$\iota$ is a bounded linear operator.
$\iota$ is not necessarily injective, because it is possible for a non-zero function
$u\in H_0^1(a,b)$ to have an $L^2_{\mu}[a,b]$-representative that has zero $L^2_{\mu}[a,b]$-norm.
To deal with this situation, we consider the subspace $\mathcal N$ of $H_0^1(a,b)$ defined
as
$$
\mathcal N:=\{u\in H_0^1(a,b):\Vert \iota(u)\Vert_{\mu}=0\}.
$$
The continuity of $\iota$ implies that $\mathcal N$ is a closed
subspace of $H_0^1(a,b)$. Let $\mathcal N^\bot$ be the
orthogonal complement of $\mathcal N$ in $H_0^1(a,b)$. It is
clear that $\iota:\mathcal N^\bot\to L^2_\mu[a,b]$ is injective, and we can identify $\mathcal N^\bot$ and $\iota(\mathcal N^\bot)$. $\mathcal N^\bot$ is dense in $L^2_{\mu}[a,b]$ (see \cite{Hu-Lau-Ngai_2006}). Throughout this paper, we let
$$
\Dom\E:=\mathcal N^\bot\qquad\text{and}\qquad\Vert\cdot\Vert_{\Dom}=\Vert\cdot\Vert_{H_0^1(a,b)}.
$$
The equivalence classes represented by $u\in H_0^1(a,b)$ and $\bar u\in L_\mu^2[a,b]$ are in general different. Corollary~\ref{Cor:continous-representative} below says that the continuous representative of $u$ lies in the intersection of these two equivalence classes. We will frequently identify $\bar u$ and $u$ without mention.

\begin{prop}
Let $u\in H_0^1(a,b)$ and $\{\phi_n\}\subset C_c^{\infty}(a,b)$ such that $\phi_n\rightarrow u$ in $H_0^1(a,b).$
Then there exists a subsequence $\{\phi_{n_k}\}$ such that $\phi_{n_k} \rightarrow u_c$ everywhere in $[a,b]$, where $u_c$ is the continuous representative of the equivalence class of $u$ in $H_0^1(a,b).$
\end{prop}

\begin{proof}
 Let  $\{\phi_{n_k}\}$ be a subsequence converging pointwise Lebesgue a.e. to $u_c$ on $(a,b)$. Let $x \in (a,b)$ and $\epsilon>0$ be arbitrary. First, since $\phi_n$ is convergent, there exists $C>0$ such that
\begin{equation}\label{equality1}
\left\|\phi_n\right\|_{\Dom{\E}}\leq C\quad  \text{ for all }n \in \N.
\end{equation}
Next, by the continuity of $u_c$, there exists $0<\delta_{\epsilon}< (\epsilon/(3C))^2 $ such that for all $y\in[a,b]$, with $ \left|y-x\right| < \delta_{\epsilon},$ we have
\begin{equation}\label{equality2}
\left|u_c(x)-u_c(y)\right|< \epsilon/3.
\end{equation}
Hence,
\begin{equation}\label{equality3}
\left|\phi_{n_k}(x)-u_c(x)\right| \leq \left|\phi_{n_k}(x)-\phi_{n_k}(y)\right| +\left|\phi_{n_k}(y)-u_c(y)\right| +\left|u_c(y)-u_c(x)\right|.
\end{equation}
The first term can be estimated by using \eqref{equality1} as follows:
\begin{equation}\label{equality4}
\begin{split}
\left|\phi_{n_k}(x)-\phi_{n_k}(y)\right|=&\left|\int_y^x \phi'_{n_k}(s)\, ds\right| \leq \Big(\int_y^x \left|\phi'_{n_k}(s)\right|^2\, ds\Big)^{1/2}\left|x-y\right|^{1/2}\\ \leq &\left\|\phi_{n_k}\right\|_{\Dom{\E}}\left|x-y\right|^{1/2} \leq C\cdot{\epsilon}/{(3C)}
\leq \epsilon/3.
\end{split}
\end{equation}
Substituting \eqref{equality2} and  \eqref{equality4} into  \eqref{equality3}, we get
$$
\left|\phi_{n_k}(x)-u_c(x)\right| \leq \epsilon/3 +\left|\phi_{n_k}(y)-u_c(y)\right| +\epsilon/3
$$
for all $y\in(x-\delta_{\epsilon},x+\delta_{\epsilon})$.

Last, let $y\in (a,b) $ satisfy $\lim_{k\rightarrow \infty}\phi_{n_k}(y)=u_c(y).$ Then, for all $k$ sufficient large,
we have  $\left|\phi_{n_k}(y)-u_c(y)\right| < \epsilon/3$, and hence $\left|\phi_{n_k}(x)-u_c(x)\right|< \epsilon.$ Thus,  $\lim_{k\rightarrow \infty}\phi_{n_k}(x)=u_c(x)$ for all $x\in[a,b]$.
\end{proof}

\begin{coro}\label{Cor:continous-representative}
Let $u \in H_0^1(a,b)$ and let $\bar{u}$ be its unique $L^2_{\mu}[a,b]$ representative. Then $u_c$ lies in the equivalence class of $\bar u$ in $L_\mu^2[a,b]$.
\end{coro}

\begin{coro}
If $\,{\rm supp}(\mu)=[a,b]$, then $\iota:H_0^1(a,b) \rightarrow L^2_{\mu}[a,b]$ is injective. Consequently, $\Dom\E =H_0^1(a,b)$.
\end{coro}

\begin{proof}
Let $u \in H_0^1(a,b)$ such that $\iota(u)=0$. Then we have $\bar{u}=u_c=0$ in $L^2_{\mu}[a,b]$.
Since ${\rm supp}(u)=[a,b],$ we have $u_c\equiv 0$ on $[a,b]$.
Thus, $u=0$ Lebesgue a.e. on $[a,b].$
\end{proof}

For a function $\varphi:(a,b)\to\R$, we let $\varphi'$ denote both its classical and weak derivatives. If $u\in L^2(0,T; X)$, where $X$ is $H^1_0(a,b)$, or $L^2_{\mu}[a,b]$ etc., then for each fixed $t$ we denote by $u_x(x,t)$ (or $\nabla u$) the classical or weak derivative of $u$ with respect to $x$.

The spaces $\Dom\E, L_\mu^2[a,b], (\Dom\E)'$ form a \textit{Gelfand triple}:
$$
\Dom\E\ \hookrightarrow\  L_\mu^2[a,b]\cong (L_\mu^2[a,b])'\ \hookrightarrow\ (\Dom\E)',
$$
where we identify $L_\mu^2[a,b]$ with $(L_\mu^2[a,b])'$. The embedding $L_\mu^2[a,b]\hookrightarrow (\Dom\E)'$ is given by
$$
w\in L_\mu^2[a,b]\mapsto (w,\cdot)_\mu\in(L_\mu^2[a,b])'\subset(\Dom\E)'.
$$

We define weak solution of the IBVP (\ref{eq:IBVP}) (see, e.g., \cite{Evans_2010, Wloka_1987}).

\begin{defi}\label{D:weak-solution}
Let $g \in \Dom\E,$ $h \in L^2_{\mu}[a,b]$, and $f \in L^2(0,T;\Dom\E)$.
A function $u \in L^2(0,T;\Dom\E)$, with
$u_t \in L^2(0,T;L^2_{\mu}[a,b] )$ and $u_{tt} \in
L^2(0,T;(\Dom\E)')$ is a {\emph{weak solution}} of
IBVP (\ref{eq:IBVP})
if the following conditions are satisfied:
\begin{enumerate}
\item[(i)] $\langle u_{tt},v \rangle + \E(u,v) = (f,v)_{\mu}$  for each $v \in \Dom\E$ and Lebesgue a.e. $t\in[0,T]$;
\item[(ii)] $u(x,0)=g(x)$ and $u_t(x,0)=h(x)$ for all $x\in[a,b]$.
\end{enumerate}
Here $\langle\cdot,\cdot\rangle$ denotes the pairing between $(\Dom\E)'$ and $\Dom\E$.
\end{defi}

\begin{rmk}\label{rmk:weak-derivative} (a) In (i) above, if $u_{tt}\in\Dom\E$ or $u_{tt}\in L_\mu^2[a,b]$, then
$\langle u_{tt},v \rangle=(u_{tt},v)_\mu$ as in the definition of Gelfand triple.

(b) Given the Gelfand triple $\Dom\E\hookrightarrow L^2_\mu[a,b]\hookrightarrow(\Dom\E)'$,
for $u\in L^2(0,T;\Dom\E)$ we also have $u\in L^2(0,T; L_\mu^2[a,b])$ and $u\in L^2(0,T; (\Dom\E)')$
and thus $u\in L^1(0,T; L_\mu^2[a,b])$ and $u\in L^1(0,T; (\Dom\E)')$. Therefore, it makes sense to require that $u$ has weak derivatives $u_{t}\in L^1(0,T ;L^2_\mu[a,b])$ and $u_{tt}\in L^1(0,T ;(\Dom\E)')$
and to require in addition that $u_t\in L^2(0,T ;L_\mu^2[a,b])$ and $u_{tt}\in L^2(0,T ;(\Dom\E)')$.
\end{rmk}

Let $V,H$ be Hilbert spaces, where $V$ is separable. Assume that the embedding $V \hookrightarrow H$ is continuous, injective, and dense so that
$$
V\hookrightarrow H\hookrightarrow V'
$$
form a Gelfand triple (see \cite[Section 17]{Wloka_1987}). Let $0<T<\infty$, and assume that for $t\in[0,T]$, $a(t,\varphi,\psi)$ is a continuous sesquilinear form, i.e.,
\begin{equation}\label{E:Wloka_condition1}
|a(t,\varphi,\psi)|\le c\Vert\varphi\Vert_V\Vert\psi\Vert_V,\quad\forall\varphi,\psi\in V,
\end{equation}
where $c>0$ is a constant independent of $t$. Then there exists a representation operator
$$
L(t):V\to V',
$$
such that for each $t$, $L(t)$ is linear and continuous, with
$$
a(t;\varphi,\psi)=(L(t)\varphi,\psi)_H.
$$

Assume that for all $\varphi,\psi\in V$ the function $t\mapsto a(t;\varphi,\psi)$ is continuously differentiable for $t\in[0,T]$, i.e.,
\begin{equation}\label{E:Wloka_condition2}
a(t;\varphi,\psi)\in C^1[0,T],\quad\forall\varphi,\psi\in V,
\end{equation}
where
$$
\Big|\frac{d}{dt}a(t;\varphi,\psi)\Big|\le c\Vert\varphi\Vert_V\Vert\psi\Vert_V,\quad\forall t\in[0,T],
$$
where $c$ is independent of $t$.

Assume further that $a(t;\varphi,\psi)$ is \textit{antisymmetric}, i.e.,
\begin{equation}\label{E:Wloka_condition3}
a(t;\varphi,\psi)=\overline{a(t;\varphi,\psi)},\quad\forall\varphi,\psi\in V.
\end{equation}

Finally, assume \textit{$V$-coersion}, i.e., there exist constants $\alpha,\beta>0$ such that
\begin{equation}\label{E:Wloka_condition4}
a(t;\varphi,\varphi)+\beta\Vert\varphi\Vert_H^2\ge\alpha\Vert\varphi\Vert_V^2,\quad\forall t\in[0,T]\text{ and }\forall\varphi\in V.
\end{equation}

The proof of Theorems \ref{T:exist-unique-Hilbert-sp} and \ref{T:smoothness-Hilbert-sp} stated below can be found in \cite[Sections 29--30]{Wloka_1987}.

\begin{theo}\label{T:exist-unique-Hilbert-sp} Let $V,H$ be Hilbert spaces where $V$ is separable. Assume that the embedding $V\hookrightarrow H$ is injective, continuous, and dense so that $V\hookrightarrow H\hookrightarrow V'$ form a Gelfand triple. Assume conditions \eqref{E:Wloka_condition1}--\eqref{E:Wloka_condition4} above hold. Then for any $f\in L^2(0,T;H)$, $0<T<\infty$, and initial conditions
$$
u_0\in V,\quad u_1\in H,
$$
there exists a unique function $u(t)\in L^2(0,T;V)$, with $du/dt\in L^2(0,T;H)$, so that
\begin{equation}\label{E:HIBV_Hilbert}
\frac{d^2u}{dt^2}+L(t)u=f\quad\text{for }t\in[0,T],\qquad u(0)=u_0,\qquad \frac{du(0)}{dt}=u_1,
\end{equation}
in the sense that
$$
\Big(\frac{d^2u}{dt^2},\varphi\Big)_H+(L(t)u,\varphi)_H=(f,\varphi)_H,\quad\forall\varphi\in V.
$$
\end{theo}

\begin{defi}\label{D:Sobolev-spaces}
Let $V$ be a Hilbert space. For each integer $k\ge 0$, define the {\rm Sobolev space}
$$
W_2^k(0,T;V):=\Big\{u:(0,T)\to V \text{ measurable}:\frac{d^nu}{dt^n}\in L^2(0,T;V)\text{ for }0\le n\le k\Big\},
$$
where the differentiation is in the distributional sense. Equip $W_2^k(0,T;V)$ with the norm
$$
\Vert u\Vert_k^2:=\sum_{n=0}^k\int_0^T\Big\Vert\frac{d^nu}{dt^n}\Big\Vert_V^2\,dt.
$$
\end{defi}

The smoothness of the solution of equation \eqref{E:HIBV_Hilbert} increases with that of $f$, as shown in the theorem below.

\begin{theo}\label{T:smoothness-Hilbert-sp}
Assume the same hypotheses of Theorem~\ref{T:exist-unique-Hilbert-sp} and assume that $a(\varphi,\psi)$ and $L$ are independent of $t$. Consider the hyperbolic equation
\begin{equation}\label{E:smoothness-Hilbert-sp-1}
\frac{d^2u}{dt^2}+Lu=f\quad\text{for }t\in(0,T),
\end{equation}
with the initial conditions
\begin{equation}\label{E:smoothness-Hilbert-sp-2}
u(0)=u_0,\qquad \frac{du(0)}{dt}=u_1.
\end{equation}
Assume, in addition, that
$$
f\in W_2^k(0,T;H),\quad k\ge 1,
$$
and
$$
u_0, u_1,f''(0),\dots, f^{(k-3)}(0)\in V\quad\text{and}\quad f^{(k-2)}(0)-Lf^{(k-3)}(0)\in H.
$$
Then the solution $u$ of \eqref{E:smoothness-Hilbert-sp-1} and \eqref{E:smoothness-Hilbert-sp-2} satisfies
$$
u\in W_2^{k-1}(0,T;V),\quad\frac{d^ku(t)}{dt^k}\in L^2(0,T;H),\quad\frac{d^{k+1}u(t)}{dt^{k+1}}\in L^2(0,T;V').
$$
\end{theo}

\begin{proof}[Proof of Theorem~\ref{T:main-1}]
In order to apply Theorem~\ref{T:exist-unique-Hilbert-sp}, we let $V=\Dom\E$, $H=L^2_\mu[a,b]$, and let $a(t;u,v)=\E(u,v)$, which independent of $t$.
Then for all $u,v\in\Dom\E$,
$$
\big|a(t;u,v)\big|
=\Big|\int u'v'\,dx\Big|
\le\Big(\int|u'|^2\,dx\Big)^{1/2}\Big(\int|v'|^2\,dx\Big)^{1/2}
=\Vert u\Vert_{\Dom(\E)}\Vert v\Vert_{\Dom(\E)},
$$
and thus condition \eqref{E:Wloka_condition1} holds. Also, $\E$ is bilinear. Thus, there exists a 
representation operator $L:\Dom\E\to(\Dom\E)'$ such that
$$
\E(u,v)=(Lu,v)_{L^2_\mu[a,b]}.
$$
$L=-\Delta_\mu$ on $\Dom(-\Delta_\mu)$.

Next, since $t\mapsto a(t;u,v)=\E(u,v)$ is constant in time and real valued, conditions  \eqref{E:Wloka_condition2} and \eqref{E:Wloka_condition3} clearly hold.

Lastly, for all $t\in[0,T]$ and $u\in V$,
$$
a(t;u,u)+\Vert u\Vert_{L^2_\mu[a,b]}^2\ge\E(u,u)=\Vert u\Vert_{\Dom\E}^2,
$$
and thus $\Dom\E$-coersion (condition~\eqref{E:Wloka_condition4}) holds with $\alpha=\beta=1$. Theorem~\ref{T:main-1} now follows from Theorem~\ref{T:exist-unique-Hilbert-sp}.
\end{proof}

As a consequence of Theorem~\ref{T:smoothness-Hilbert-sp}, we have the following regularity result for solutions of homogeneous wave equations in our setting.

\begin{theo}\label{T:smoothness}
Assume the same hypotheses of Theorem~\ref{T:main-1} and assume, in addition, that $h\in\Dom\E$ and $f=0$. Then the solution of the homogeneous equation \eqref{eq:IBVP} satisfies:
$$
u\in W_2^{k-1}(0,T;\Dom\E),\quad\frac{d^k u}{dt^k}\in L^2(0,T;L^2_\mu[a,b]),\quad \frac{d^{k+1} u}{dt^{k+1}}\in L^2(0,T;(\Dom\E)'),\quad k\ge 1.
$$
\end{theo}

Theorem~\ref{T:smoothness} will be used in proving Theorem~\ref{T:main-3}.

\section{The finite element method}\label{S:FEM}
\setcounter{equation}{0}

In this section, we let $f=0$ in equation \eqref{eq:IBVP} and use the finite element method to solve the homogeneous IBVP. We only consider self-similar measures $\mu$ (see \eqref{self-similarmeasure}) defined by an IFS  $\{S_i\}_{i=1}^q$ of contractive similitudes of the form
$$
S_i(x)=\rho x+b_i,\qquad i=1,\dots,q.
$$
We assume in addition that $\mu$ satisfies a family of second-order self-similar identities with
respect to an auxiliary IFS $\{T_j\}_{j=1}^{N}$ of the form \eqref{eq:IFS_T_j}. Assume also that $\{T_j\}_{j=1}^{N}$ satisfies the OSC.

For each multi-index $J=(j_1,\dots,j_m)\in\{1,\dots,N\}^m$, we let $T_J[a,b]$ be the interval
$[x_{i-1}, x_i]$, where the index $i$ is obtained directly from $J$ as follows (see \cite{Chen-Ngai_2010}):
$$
i=i(J):=(j_1-1)N^{m-1}+(j_2-1)N^{m-2}+\cdots+(j_m-1)N^0 + 1.
$$
For example, if $J=(1,\dots,1)$, then
$i(J)=1$, and if $J=(N,\dots,N)$, then $i(J)=N^m$. We call $T_J[a,b]$ a \textit{level-$m$ subinterval}.
It follows that
\begin{equation}\label{eq:T_J-formula}
    T_{J_i}[a,b]:=T_J[a,b]=[x_{i-1},x_i] \quad\mathrm{and}\quad
    T_{J_i}(x):=T_J(x)=(x_i-x_{i-1})\frac{x-a}{b-a}+x_{i-1}.
\end{equation}

We apply the finite element method to approximate
the weak solution $u(x,t)$ satisfying \eqref{eq:PDE-weakform-main} by
\begin{equation}\label{eq:u-m-expansion}
u^m(x,t)=\sum_{j=0}^{N^m}\beta_j(t)\phi_j(x),
\end{equation}
where for $j=0,1,\dots,N^m$, $\beta_j(t)=\beta_{m,j}(t)$ are functions to be determined, and $\phi_j(x):=\phi_{m,j}(x)$ are the standard piecewise linear \textit{finite element basis functions} (also called \textit{tent functions}) defined as
\begin{equation}\label{eq:FEM-basis-fn}
\phi_j(x):=\left\{
        \begin{array}{ll}
        \frac{x-x_{j-1}}{x_j-x_{j-1}} & \mathrm{if}\quad x_{j-1}
            \leq x \leq x_j, \quad j = 1,\ldots, N^m\\[7pt]
        \frac{x-x_{j+1}}{x_j-x_{j+1}} & \mathrm{if}\quad x_j
            \leq x \leq x_{j+1}, \quad j = 0,\ldots, N^m-1\\[7pt]
        0 & \mathrm{otherwise}.
        \end{array} \right.
\end{equation}

We require $u^m(x,t)$ to satisfy the integral form of the homogeneous wave equation
\begin{equation}\label{eq:PDE-weakform}
\int^{b}_{a}
u^m_{tt}(x,t)\phi_i(x) \,d\mu=-\int^{b}_{a} \nabla u^m(x,t) \phi'_i(x) \, dx,\quad\text{for } i=0,1,\dots, N^m,
\end{equation}
where $u^m_{tt}:=(u^m)_{tt}$.

As $u^m(a,t)=u^m(b,t)=0$ we have $\beta_0(t)=\beta_{N^m}(t)=0$. Using this and substituting (\ref{eq:u-m-expansion}) into (\ref{eq:PDE-weakform}) gives
\begin{equation}\label{eq:u-m-expansion2}
\sum_{j=1}^{N^m-1}\beta''_j(t)\int^{b}_{a}
\phi_i(x)\phi_j(x)\,d\mu=-\sum_{j=1}^{N^m-1}\beta_j(t)\int^{b}_{a}\phi'_i(x)\phi'_j(x)\, dx,\quad 1\le i\le N^m-1.
\end{equation}

We define the \textit{mass matrix}
$\mathbf M=(M_{ij})$ and \textit{stiffness matrix} $\mathbf K=(K_{ij})$ respectively as
\begin{equation}\label{eq:mass-stiffness}
M_{ij}=\int^{b}_{a}\phi_i(x) \phi_j(x)\,d\mu,\qquad
K_{ij}= -\int^{b}_{a}\phi'_i(x)\phi'_j(x)\, dx,\qquad 1\le i,j\le N^m-1.
\end{equation}
Both $\mathbf{M}$ and $\mathbf{K}$ are tridiagonal and of order $(N^m-1)\times (N^m-1)$.
Let
 \begin{equation*}
 \mathbf{w}(t)=:\begin{bmatrix}
 w_1(t)\\
 \vdots \\
w_{N^m-1}(t)
 \end{bmatrix}=\begin{bmatrix}
 \beta_1(t)\\
 \vdots \\
 \beta_{N^m-1}(t)
 \end{bmatrix}.
 \end{equation*}
be a vector-valued function. Then (\ref{eq:u-m-expansion2}) can be put in a matrix form as
  \begin{equation}\label{eq:matrixform2}
  \mathbf{M}\mathbf{w''}=-\mathbf{K} \mathbf{w}.
   \end{equation}
This gives us a system of second-order linear ODEs with constant coefficients.
We need two initial conditions. The initial
condition $u(x,0)= g(x)$ for {$a\le x\le b$}
can be approximated by its linear interpolant:
    \begin{equation*}
  \tilde{g}(x)= \sum_{i=1}^{N^m-1}g(x_i)\phi_i(x).
   \end{equation*}
Therefore, we set
 \begin{equation*}
w_i(0)= g(x_i)\qquad\text{and}\qquad
w'_i(0)= h(x_i).
   \end{equation*}
These lead to the initial conditions
\begin{equation}\label{eq:initial_cond}
 \mathbf{w}(0)=\mathbf{w}_0=\begin{bmatrix}
 g(x_1)\\
 \vdots \\
 g(x_{N^m-1})
\end{bmatrix}
 ,\qquad \mathbf{w'}(0)=\mathbf w'_0=
 \begin{bmatrix}
 h(x_1)\\
 \vdots \\
 h(x_{N^m-1})
\end{bmatrix}.
\end{equation}
Consequently, we get the linear system
\begin{equation}\label{eq:IVP}
\left\{
        \begin{array}{l}
\mathbf{M} \dfrac{d^2\mathbf{w}}{dt^2}=-\mathbf{K}\mathbf{w},\quad t>0 \\ \\
\mathbf{w}(0)=\mathbf{w}_0, \quad \mathbf w'(0)=\mathbf w'_0.\\
\end{array} \right.
\end{equation}

We describe how to compute $\mathbf M$; the matrix $\mathbf K$ can be computed directly. By using the definition of the $\phi_i$'s and \eqref{eq:T_J-formula}, we have
\begin{equation}\label{eq:M_entry}
\begin{aligned}
&M_{i,i}=\frac{1}{(b-a)^2}\Big(\int^{b}_{a}(x-a)^2\,d\mu\circ T_{J_i}+\int^{b}_{a}(b-x)^2\,d\mu\circ T_{J_{i+1}}\Big),\quad 1\le i\le N^m-1,\\
&M_{i,i-1}=\frac{1}{(b-a)^2}\int^{b}_{a}(x-a)(b-x)\,d\mu\circ T_{J_{i}},\quad 2\le i\le N^m-1,\\
&M_{i,i+1}=
\frac{1}{(b-a)^2}\int^{b}_{a}(x-a)(b-x)\,d\mu\circ T_{J_{i+1}},\quad 1\le i\le N^m-2.
\end{aligned}
\end{equation}

Define
\begin{equation}\label{eq:I_J}
\I_{k,j}:=\int_a^bx^k\,d\mu\circ T_j,\qquad \J_{k,j}:=\int_{T_j[a,b]}x^k\,d\mu,\qquad k=0,1,2,\ j=1,\dots, N.
\end{equation}
We will regard the $\I_{k,j}$ and $\J_{k,j}$ as known constants. In fact, for all examples we study, they can be computed exactly (see Section~\ref{S:Frac}). A sufficient condition for computing them explicitly is given in \cite{Chen-Ngai_2010}.

\begin{lemma}\label{P:M_determine_integrals}
The matrix $\mathbf M$ is completely determined by the integrals $\I_{k,j}$, or equivalently, $\J_{k,j}$, where $k=0,1,2$ and $j=1,\dots, N$.
\end{lemma}

\begin{proof}
For $J=(j_1,\dots,j_m)\in\{1,\dots,N\}^m$, iterating \eqref{2ndorderidentity} shows that for any Borel subset  $A\subseteq{\rm supp}(\mu)$,
\begin{equation}\label{E:iterated_2nd_order_id}
\mu(T_JA)=c_J\begin{bmatrix}\mu(T_1A)\\
\vdots\\
\mu(T_NA)\end{bmatrix},
\end{equation}
where $c_J:=[c_J^1,\dots,c_J^N]:=\emph{\textbf{e}}_{j_1}M_{j_2}\cdots M_{j_m}$. That is,
\begin{equation}\label{eq:mu_T_J}
\mu(T_JA)=\sum_{j=1}^Nc_J^i\mu(T_jA).
\end{equation}
In view of the fact that $\mathbf M$ is tridiagonal, and the expressions for $M_{i,i}$, $M_{i,i-1}$, and  $M_{i,i+1}$, the entries of $\mathbf M$ are completely determined by the integrals
$$
\int_a^bx^k\,d\mu\circ T_J,\quad k=0,1,2,\quad J\in\{1,\dots, N\}^m,
$$
which, by virtue of \eqref{eq:mu_T_J}, can be written as
$$
\sum_{j=1}^Nc_J^i\int_a^bx^k\,d\mu\circ T_j.
$$
This proves that $\mathbf M$ is determined by the $\I_{k,j}$. Lastly, since
$$
\int_a^bx^k\,d\mu\circ T_j=\int_{T_j[a,b]}(T_j^{-1}x)^k\,d\mu\quad\text{and}\quad
\int_{T_j[a,b]}x^k\,d\mu=\int_a^b(T_jx)^k\,d\mu\circ T_j,
$$
$\mathbf M$ is also determined by the $\J_{k,j}$.
\end{proof}

The system in \eqref{eq:IVP} has a unique solution if $\mathbf{M}$ is invertible.

\begin{prop}\label{P:M-invertible} Assume that ${\rm supp}(\mu)=[a,b]$. Then the mass matrix $\mathbf M$ is invertible. Consequently, (\ref{eq:IVP}) has a unique solution $\mathbf w(t)$; moreover, ${\beta_j(t)}\in C^2(0,T)$ for $j=1,\dots,N^m-1$.
\end{prop}

\begin{proof}
If the mass matrix $\mathbf M$ is not invertible, then there exists a nonzero piece-wise linear function
with zero $L^2_\mu$ norm, which implies that the measure $\mu$ does not have a full support.
\end{proof}

\begin{proof}[{\rm \textit{Proof of Theorem \ref{T:main-2}}}]
This follows by combining the derivations above, Lemma~\ref{P:M_determine_integrals}, and Proposition~\ref{P:M-invertible}.
\end{proof}

We now give another sufficient condition for the matrix $\mathbf{M}$ to be invertible.
If we define
\begin{equation}\label{eq:invertibility-condition1}
\begin{aligned}
&p_1(x):=\frac{(x-a)^2}{(b-a)^2},\qquad &&p_2(x):=\frac{(x-a)(2x-a-b)}{(b-a)^2}, \\
&p_3(x):=\frac{(b-x)(a+b-2x)}{(b-a)^2},\qquad &&p_4(x):=\frac{(b-x)^2}{(b-a)^2},
\end{aligned}
\end{equation}
then
\begin{equation}\label{eq:invertibility-condition2}
\begin{aligned}
&M_{1,1}-M_{1,2}
=\int^{b}_{a} p_1\,d\mu\circ T_{J_{1}}
+\int^{b}_{a} p_2\,d\mu\circ T_{J_{2}},\\
&M_{i,i}-M_{i,i-1}-M_{i,i+1}
=\int^{b}_{a} p_2\,d\mu\circ T_{J_{i}}
+\int^{b}_{a} p_3\,d\mu\circ T_{J_{i+1}},\quad 2\le i\le N^m-2,\\
&M_{N^m-1,N^m-1}-M_{N^m-1,N^m-2}
=\int^{b}_{a} p_3\,d\mu\circ T_{J_{N^m-1}}
+\int^{b}_{a} p_4\,d\mu\circ T_{J_{N^m}}.\\
\end{aligned}
\end{equation}

We recall that an $n\times n$ complex matrix $A=(a_{ij})$ is \textit{strictly diagonally dominant} if
\begin{equation}\label{eq:3.1}
|a_{ii}|>\sum_{j=1,j\ne i}^n|a_{ij}|\quad\text{for }1\le i\le n.
\end{equation}
It is well known that any $n\times n$ strictly diagonally dominant complex matrix is invertible (see e.g., \cite{Varga_2000}).

\begin{prop}\label{P:FEM}
Let $\mathbf M$ be the mass matrix defined in (\ref{eq:mass-stiffness}) and $p_i, i=1,\dots,4$, be defined as in \eqref{eq:invertibility-condition1}. Assume that
\begin{equation*}
\begin{aligned}
&\int^{b}_{a} p_1\,d\mu\circ T_{J_{1}}+\int^{b}_{a} p_2\,d\mu\circ T_{J_{2}}>0,\quad\int^{b}_{a} p_3\,d\mu\circ T_{J_{N^m-1}}+\int^{b}_{a} p_4\,d\mu\circ T_{J_{N^m}}>0,\quad\text{and}\\
&\int^{b}_{a} p_2\,d\mu\circ T_{J_{i}}+\int^{b}_{a} p_3\,d\mu\circ T_{J_{i+1}}>0,\quad\text{for all}\quad 2\le i\le N^m-2.
\end{aligned}
\end{equation*}
Then $\mathbf M$ is strictly diagonally dominant and thus invertible. Hence the same conclusions of Proposition~\ref{P:M-invertible} hold.
\end{prop}

For the infinite Bernoulli convolution associated with the golden ratio, as well as the 3-fold convolution of the Cantor measure (see Section~\ref{S:Frac}), we can verify that $\mathbf M$ is strictly diagonally dominant; we omit the details.

Next, we discuss the solution of the linear system \eqref{eq:matrixform2}. We let $\mathbf{w}_n:=\mathbf{w}(t_n)$, $n\geq -1$, and use the {\it central difference method} to solve the IVP (\ref{eq:IVP}). (The value of $\mathbf{w}_{-1}$ is defined  below.)

We approximate the derivatives as follows:
\begin{equation}\label{eq:matrixform3}
\frac{d^2\mathbf{w}(t_n)}{dt^2}
\approx\frac{\mathbf{w}_{n+1}-2\mathbf{w}_n+\mathbf{w}_{n-1}}{(\Delta t)^2}\qquad\text{and}\qquad
\mathbf{w}{'}(t_n)\approx \frac{\mathbf{w}_{n+1}-\mathbf{w}_{n-1}}{2\Delta t}.
\end{equation}
Substituting (\ref{eq:matrixform3}) into  (\ref{eq:matrixform2}) yields
\begin{equation*}
\frac{\mathbf{w}_{n+1}-2\mathbf{w}_n+\mathbf{w}_{n-1}}{(\Delta t)^2}=-\mathbf{M^{-1}}\mathbf K\mathbf{w}_n,\quad\text{i.e.,}\quad\mathbf{w}_{n+1}=(2\mathbf{I}-(\Delta t)^2\mathbf{M}^{-1}\mathbf{K})\mathbf{w}_n-\mathbf{w}_{n-1}.
\end{equation*}
Moreover, using
\begin{equation*}
 \mathbf{w}_{1}=(2\mathbf{I}-(\Delta t)^2\mathbf{M}^{-1}\mathbf{K})\mathbf{w}_0-\mathbf{w}_{-1}\qquad\text{and}\qquad
 \mathbf w'_0= \frac{\mathbf{w}_{1}-\mathbf{w}_{-1}}{2\Delta t},
    \end{equation*}
we get
\begin{equation}\label{eq:w_1}
 \mathbf{w}_{1}=\Big(\mathbf{I}-\frac{(\Delta t)^2}{2}\mathbf{M}^{-1}\mathbf{K}\Big)\mathbf{w}_0+(\Delta t)\mathbf w'_0.
\end{equation}
Therefore, equation (\ref{eq:matrixform2}) becomes:
\begin{equation}\label{eq:IVP-CDM}
\left\{
        \begin{aligned}
\mathbf{w}_{n+1}&=(2\mathbf{I}-(\Delta t)^2\mathbf{M}^{-1}\mathbf{K})\mathbf{w}_n-\mathbf{w}_{n-1},\quad n=1,2,\dots \\
\mathbf{w}_0&=\mathbf{w}(t_0)=\mathbf{w}(0) \\
\mathbf{w}_1&=\mathbf{w}(t_1)=\Big(\mathbf{I}-\frac{(\Delta t)^2}{2}\mathbf{M}^{-1}\mathbf{K}\Big)
\mathbf{w}_0+(\Delta t)\mathbf w'_0\\
t_n&=n\Delta t.
\end{aligned} \right.
\end{equation}

To solve this system, we fix $\Delta t$ and substitute the initial conditions $\mathbf w_0$ and $\mathbf w_0'$ from \eqref{eq:initial_cond} into \eqref{eq:w_1} to get $\mathbf w_1$.
Then substitute $\mathbf w_0$
and $\mathbf w_1$ into the first equation in \eqref{eq:IVP-CDM} to find $\mathbf w_2$. $\mathbf w_{n+1}$ can then be computed recursively.

\section{Fractal measures defined by iterated function systems}\label{S:Frac}
\setcounter{equation}{0}

In this section, we solve the homogeneous IBVP \eqref{eq:IBVP} numerically for three different measures, namely, a weighted Bernoulli-type measure, the infinite Bernoulli convolution associated with the golden ratio, and the 3-fold convolution of the Cantor measure.
The first one is defined by a p.c.f. IFS, while the second and third are defined by IFSs with overlaps.

We assume the same hypotheses of Section~\ref{S:FEM}.
In order to solve \eqref{eq:IVP} or \eqref{eq:IVP-CDM}, we need to compute the matrix $\mathbf{M}$ (the matrix $\mathbf{K}$ can be computed easily). According to Lemma~\ref{P:M_determine_integrals}, it suffices to compute
the integrals $\I_{k,j}$, $k=0,1,2$, $j=1,\dots, N$, as defined in \eqref{eq:I_J}. We find the exact values of these integrals for the measures in this section. The following integration formula will be used repeatedly: for any continuous function $\varphi$ on $\text{supp}(\mu)=[a,b]$,
\begin{equation}\label{eq:int_wrt_ssm}
\int_a^b \varphi\,d\mu = \sum_{i=1}^q p_i\int_a^b\varphi\circ S_i\,d\mu.
\end{equation}
By substituting the values of $\I_{k,j}$ into \eqref{eq:M_entry}, we obtain the matrix $\mathbf M$.
This allows us to solve equation \eqref{eq:IVP-CDM}.

\subsection{Weighted Bernoulli-type 
measure}
A weighted Bernoulli-type 
measure $\mu$ is defined by the IFS
$$
S_1(x)=\frac{1}{2} x,\qquad S_2(x)=\frac{1}{2}x+\frac{1}{2},
$$
together with probability weights $p,1-p$. Thus,
$$ \mu= p\mu\circ S_1^{-1}
          + (1- p)\mu\circ S_2^{-1}.
$$

For any Borel subset $A \subseteq [0,1],$
we have:
\begin{equation*}
\left[ \begin{array}{c}
        \mu(S_1S_iA)\\
        \mu(S_2S_iA)\\
\end{array}
\right] =M_i\left[\begin{array}{c}
        \mu(S_1A)\\
        \mu(S_2A)\\
\end{array}
\right], \quad i=1,2,
\end{equation*}
where
\begin{displaymath}
                M_1=\left[\begin{array}{cc}
                 p      & 0 \\
                 0 &    p
                \end{array}
                \right]\qquad\text{and}\qquad
                M_2=\left[\begin{array}{cc}
                 1-p      & 0\\
                 0 & 1-p \\
                 \end{array}
                \right].\
\end{displaymath}

Let $J=j_1j_2\cdots j_m$, $j_i=1$ or $2$. Then
      \begin{equation*}
          \mu(S_JA) = c_J \left[\begin{array}{c}
              \mu(S_1A)\\
              \mu(S_2A)
          \end{array}\right],\quad\text{where}\quad c_J=\emph{\textbf{e}}_{j_1}M_{j_2}\cdots M_{j_m}=(c_J^1,c_J^2).
      \end{equation*}
Since the IFS satisfies the open set condition, it is straightforward to evaluate the integrals $\I_{k,j}$; we omit the details. In view of \cite{Bird-Ngai-Teplyaev_2003}, we choose the weight
$p=2-\sqrt{3}$ in Figure~\ref{fig:weighted_Lebesgue_measure1}.

\begin{figure}[h]
\centering \mbox{
      {\epsfig{file=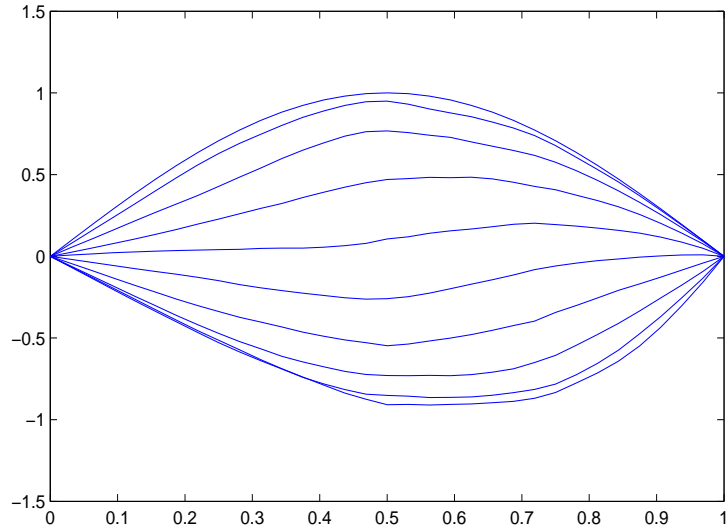,
      width=0.70\textwidth, height=0.35\textheight}}
      }
\caption{The weighted Bernoulli-type 
measure associated with the weights $p=2-\sqrt{3}$ and
$1-p= \sqrt{3}-1$. The initial data $g=\sin(\pi x)$ and $h=0$ are used, and the time step $\Delta t$ in equation~\eqref{eq:IVP-CDM} is taken to be $0.001$. From top to bottom, the values of $t$ are $0.0, 0.1, 0.2, 0.3, 0.4, 0.5, 0.6, 0.7, 0.8, 0.9$. Animations for this and other graphs in the paper are created and uploaded to the webpage \url{http://homepages.uconn.edu/fractals/wave/.}}
 \label{fig:weighted_Lebesgue_measure1}
\end{figure}

\subsection{Infinite Bernoulli convolution associated with the golden ratio}

$\linebreak$
The infinite Bernoulli convolution associated with the golden ratio
is defined by the IFS
$$
S_1(x)=\rho x,\qquad S_2(x)=\rho x+(1-\rho),\qquad \rho=\frac{\sqrt
5-1}{2}.
$$
\begin{center}
  \begin{picture}(130,65)(0,15)
      \unitlength=.5mm

      \put(0,50){\line(1,0){100}}
      \put(0,50){\line(0,1){1}}
      \put(33.3,50){\line(0,1){1}}
      \put(66.7,50){\line(0,1){1}}
      \put(100,50){\line(0,1){1}}
      \put(-1,42){$0$}
      \put(99,42){$1$}

      \put(33.3,41){\vector(-1,-2){7}}
      \put(19,35){$S_1$}

      \put(66.7,41){\vector(1,-2){7}}
      \put(72,35){$S_2$}

      \put(0,20){\line(1,0){60.3}}
      \put(0,20){\line(0,1){1}}
      \put(60.3,20){\line(0,1){1}}

      \put(39.7,14){\line(1,0){60.3}}
      \put(39.7,14){\line(0,1){1}}
      \put(100,14){\line(0,1){1}}

      \put(30.2,7){$1-\rho$}
      \put(58.7,7){$\rho$}
      \put(0,7){$0$}
      \put(100,7){$1$}

  \end{picture}
\end{center}

\vspace{2ex}
For each $0<p<1 $, we call the corresponding self-similar measure
$$
\mu=p\mu\circ S_1^{-1}+(1-p)\mu\circ S_2^{-1},
$$
a \textit{weighted infinite Bernoulli convolution associated with the
golden ratio}. If $p=1/2$, we get the classical one.

The measure $\mu_p$ satisfies a family of second-order identities. This was
first pointed out by Strichartz et al.
\cite{Strichartz-Taylor-Zhang_1995}. Define
$$
T_1(x)= \rho^2 x,\qquad T_2(x)= \rho^3 x+\rho^2,\qquad T_3(x)= \rho^2
x+\rho.
$$
Then $\mu$ satisfies the following second-order identities (see \cite{Lau-Ngai_2000}): for any Borel subset $A \subseteq [0,1],$
\begin{equation*}
\left[ \begin{array}{c}
        \mu(T_1T_iA)\\
        \mu(T_2T_iA)\\
        \mu(T_3T_iA)
\end{array}
\right] =M_i\left[\begin{array}{c}
        \mu(T_1A)\\
        \mu(T_2A)\\
        \mu(T_3A)
\end{array}
\right], \quad i=1,2,3,
\end{equation*}
where $M_1, M_2, M_3$ are, respectively,
$$
\begin{bmatrix}
p^2      &  0     & 0\\
                 (1-p)p^2 & (1-p)p & 0\\
                 0        &  1-p   & 0
\end{bmatrix},\quad
\begin{bmatrix}
0 & p^2     & 0\\
                 0 & (1-p)p  & 0\\
                 0 & (1-p)^2 & 0
\end{bmatrix},\quad
\begin{bmatrix}
 0 &  p     & 0       \\
                 0 & (1-p)p & (1-p)^2p\\
                 0 &  0     & (1-p)^2
\end{bmatrix}.
$$
We can make use of this to compute the measure of suitable subintervals of $[0,1]$. In fact, if we let $J=j_1\cdots j_m$, $j_i=1,2$ or $3$, then for any Borel subset $A\subseteq [0,1]$,
      \begin{equation*}
          \mu(T_JA) = c_J \left[\begin{array}{c}
              \mu(T_1A)\\
              \mu(T_2A)\\
              \mu(T_3A)
          \end{array}\right],\quad\text{where}\quad c_J=\emph{\textbf{e}}_{j_1}M_{j_2}\cdots M_{j_m}=(c_J^1,c_J^2,c_J^3).
      \end{equation*}
Moreover, by using \eqref{eq:int_wrt_ssm}
we can evaluate the integrals $\I_{k,j}$ in \eqref{eq:I_J}. For  $p=1/2$, the results are summarized below:
\begin{equation}\label{eq:measure-Bernoulli}
\begin{aligned}
&   \int_0^1      \,d\mu\circ T_{1}=\frac{1}{3}\
& & \int_0^1      \,d\mu\circ T_{2}=\frac{1}{3}\
& & \int_0^1      \,d\mu\circ T_{3}=\frac{1}{3}\\
&   \int_0^1 x    \,d\mu\circ T_{1}=\frac{1}{6(3\rho-1)}\
& & \int_0^1 x    \,d\mu\circ T_{2}=\frac{1}{6}\
& & \int_0^1 x    \,d\mu\circ T_{3}=\frac{1}{6(3\rho^{2}+3)}\\
&   \int_0^1 x^{2}\,d\mu\circ T_{1}=\frac{5\rho+4}{6(\rho+8)}\
& & \int_0^1 x^{2}\,d\mu\circ T_{2}=\frac{\rho+5}{6(\rho+8)}\
& & \int_0^1 x^{2}\,d\mu\circ T_{3}=\frac{2-\rho}{6(\rho+8)}.
\end{aligned}
\end{equation}
We can thus calculate the entries of the mass matrix $\mathbf{M}$ and solve the linear system \eqref{eq:IVP}. The result is shown in Figure~\ref{fig:Infinite Bernoulli convolution}.
\begin{figure}[h]
\centering \mbox{
      {\epsfig{file=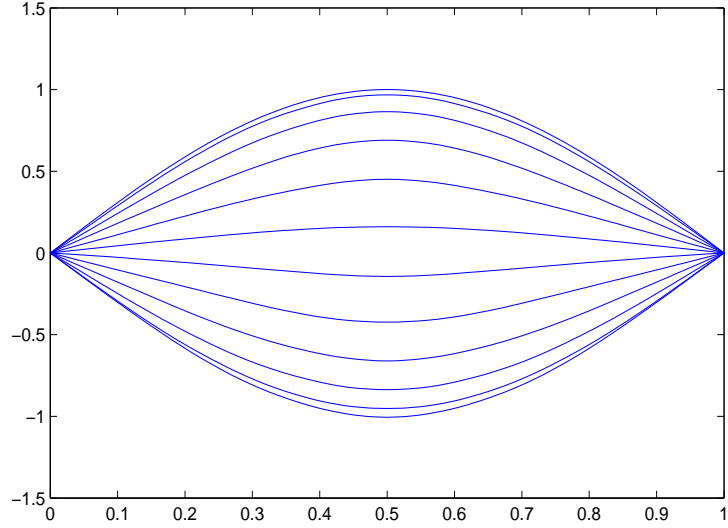,
      width=0.70\textwidth, height=0.35\textheight}}
      }
\caption{Infinite Bernoulli
convolution associated with the golden ratio. The initial data $g=\sin(\pi x)$ and $h=0$ are used. The time step $\Delta t$ in equation \eqref{eq:IVP-CDM} is taken to be $0.001$.
From top to bottom, the values of $t$ are $0.0, 0.1, 0.2, 0.3, 0.4, 0,5, 0.6, 0,7, 0.8, 0.9, 1.0, 1.1$. }
 \label{fig:Infinite Bernoulli
convolution}
\end{figure}

\subsection{3-fold convolution of the Cantor measure}

The 3-fold convolution of the Cantor measure $\mu$ also satisfies a family of
second-order identities. It is defined by the IFS
\begin{equation*}
    S_i(x)=\frac{1}{3}x+\frac{2}{3}(i-1),\quad \mathrm{for} \quad
    i=1,2,3,4,
\end{equation*}
which does not satisfy the OSC.
\begin{center}
  \begin{picture}(130,65)(0,15)
      \unitlength=.5mm

      \put(0,50){\line(1,0){100}}
      \put(0,50){\line(0,1){1}}
      \put(33.3,50){\line(0,1){1}}
      \put(66.7,50){\line(0,1){1}}
      \put(100,50){\line(0,1){1}}
      \put(-1,42){$0$}
      \put(99,42){$3$}

      \put(20,41){\vector(-1,-2){8}}
      \put(8,35){$S_1$}

      \put(43,41){\vector(-1,-4){4}}
      \put(33,35){$S_2$}

      \put(57,41){\vector(1,-4){4}}
      \put(60,35){$S_3$}

      \put(80,41){\vector(1,-2){8}}
      \put(85,35){$S_4$}

      \put(0,20){\line(1,0){33.3}}
      \put(0,20){\line(0,1){1}}
      \put(33.3,20){\line(0,1){1}}

      \put(22.2,17){\line(1,0){33.3}}
      \put(22.2,17){\line(0,1){1}}
      \put(55.7,17){\line(0,1){1}}

      \put(44.4,14){\line(1,0){33.3}}
      \put(77.7,14){\line(0,1){1}}
      \put(44.4,14){\line(0,1){1}}

      \put(66.7,11){\line(1,0){33.3}}
      \put(66.7,11){\line(0,1){1}}
      \put(100,11){\line(0,1){1}}

      \put(32.2,4){$1$}
      \put(65.7,4){$2$}
      \put(0,4){$0$}
      \put(100,4){$3$}

  \end{picture}
\end{center}

\vspace{3ex}
The measure $\mu$ satisfies the following self-similar identity:
$$ \mu= \frac{1}{8}\mu\circ S_1^{-1}
          + \frac{3}{8}\mu\circ S_2^{-1}
          +\frac{3}{8}\mu\circ S_3^{-1}
          + \frac{1}{8}\mu\circ S_4^{-1}
$$

Define
$$
T_1(x)= \frac{1}{3} x, \quad T_2(x)= \frac{1}{3} x+1,\quad T_3(x)=
\frac{1}{3} x+2.
$$
Then $\mu$ satisfies the following second-order identities (see \cite{Lau-Ngai_2000}): for any Borel subset
$A\subseteq[0,3]$,
        \begin{equation*}
            \left[\begin{array}{c}
                \mu(T_{1j}A)\\
                \mu(T_{2j}A)\\
                \mu(T_{3j}A)
            \end{array}\right]
            = M_j \left[\begin{array}{c}
                \mu(T_1A)\\
                \mu(T_2A)\\
                \mu(T_3A)
            \end{array}\right],
            \quad j=1,2,3,
        \end{equation*}
where the coefficient matrices $M_j$ are given by
        $$
            M_1=\frac{1}{8}\left[
                \begin{array}{ccc}
                    1&0&0\\
                    0&3&0\\
                    1&0&3
                \end{array}
                \right],\quad
            M_2=\frac{1}{8}\left[
                \begin{array}{ccc}
                    0&1&0\\
                    3&0&3\\
                    0&1&0
                \end{array}
                \right],\quad
            M_3=\frac{1}{8}\left[
                \begin{array}{ccc}
                    3&0&1\\
                    0&3&0\\
                    0&0&1
                \end{array}
                \right].
        $$
Let $J=j_1\cdots j_m$, $j_i=1,2$ or $3$. Then
      \begin{equation*}
          \mu(T_JA) = c_J \left[\begin{array}{c}
              \mu(T_1A)\\
              \mu(T_2A)\\
              \mu(T_3A)
          \end{array}\right],\quad\text{where}\quad c_J=\emph{\textbf{e}}_{j_1}M_{j_2}\cdots M_{j_m}=(c_J^1,c_J^2,c_J^3).
      \end{equation*}

The integrals $\I_{k,j}$ in \eqref{eq:I_J} are given below:
\begin{equation}\label{eq:measure-Cantor}
\begin{aligned}
&  \int_0^3       \,d \mu \circ T_{1}=\frac{1}{5}\
& &\int_0^3       \,d \mu \circ T_{2}=\frac{3}{5}\
& &\int_0^3       \,d \mu \circ T_{3}=\frac{1}{5} \\
&  \int_0^3 x     \,d \mu \circ T_{1}= \frac{27}{70}\
& &\int_0^3 x     \,d \mu \circ T_{2}=\frac{9}{10}\
& &\int_0^3 x     \,d \mu \circ T_{3}=\frac{3}{14}\\
&  \int_0^3 x^{2} \,d \mu \circ T_{1}= \frac{5517}{6440}\
& &\int_0^3 x^{2} \,d \mu \circ T_{2}=\frac{11943}{6440}\
& &\int_0^3 x^{2} \,d \mu \circ T_{3}= \frac{63}{184}.
\end{aligned}
\end{equation}

Again, using these values we can compute $\mathbf{M}$ and solve \eqref{eq:IVP} (see Figure~\ref{fig:3-fold convolution}).

\begin{figure}[h]
\centering \mbox{
      {\epsfig{file=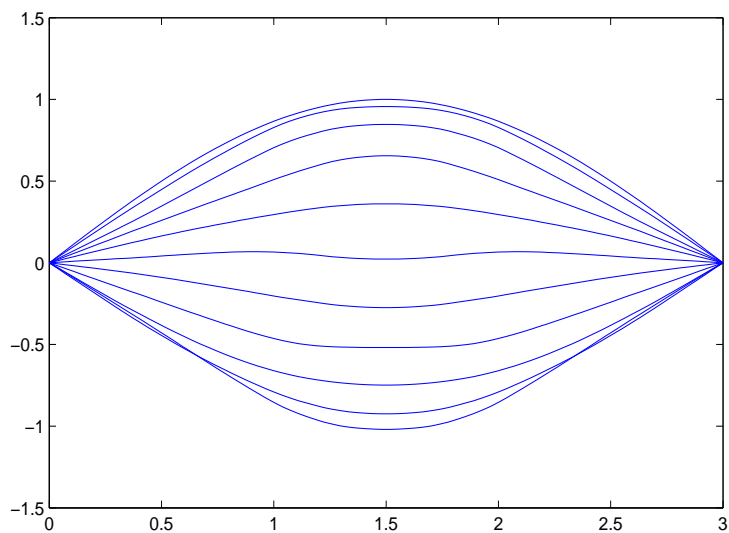,
      width=0.70\textwidth, height=0.35\textheight}}
      }
\caption{Three-fold convolution of
the Cantor measure. The initial data $g=\sin(\pi x/3)$ and $h=0$ are used, and $\Delta t=0.001$.
From top to bottom,
the values of $t$ are $0.0, 0.2, 0.4, 0.6, 0.8, 1.0, 1.2, 1.4, 1.6, 1.8, 2.0$.
}
 \label{fig:3-fold convolution}
\end{figure}

\section{Convergence of numerical approximations}\label{S:convergence}
\setcounter{figure}{0}

In this section we prove the convergence of the numerical approximations of the homogeneous IBVP \eqref{eq:IBVP}. Some of our results are obtained by modifying similar ones in \cite{Strang-Fix_1973} (see also \cite{Chen-Ngai_2010}).

We assume the same setup of Section~\ref{S:FEM} unless stated otherwise.
Let $V_m$ be the set of end-points of all the level-$m$ intervals, and arrange its elements
so that $V_m=\{x_i:i=0,1,\ldots, N^m\}$ with $x_i < x_{i+1}$ for
$i=0,1,\dots,N^m-1$, $x_0=a$ and $x_{N^m}=b$. Let $S^m$ be the space of continuous
piecewise linear functions on $[a,b]$ with nodes $V_m$, and let
$$
S_D^m:=\{u\in S^m:u(a)=u(b)=0\}
$$
be the subspace of $S^m$ consisting of functions satisfying the Dirichlet boundary
condition. Then
$$
\dim S^m=\#V_m=N^m+1\qquad\text{and}\qquad \dim S_D^m =\#V_m - 2 = N^m - 1.
$$

We choose the basis of $S^m$ consisting of the tent functions $\{\phi_i\}_{i=0}^{N^m}$
defined in (\ref{eq:FEM-basis-fn}) and choose the basis $\{\phi_i\}_{i=1}^{N^m-1}$ for $S_D^m$.

\begin{defi}\label{D:Rayleigh-Ritz}
Let $V_m$ be defined as above and $\{\phi_i\}_{i=0}^{N^m}$ be
defined as in (\ref{eq:FEM-basis-fn}). The linear map $\PP_m : \Dom\E \rightarrow S_D^m$ defined by
$$
\PP_m v:= \sum_{i=1}^{N^m-1}v(x_i)\phi_i(x),\quad v\in \Dom\E,
$$
is called the {\rm Rayleigh-Ritz projection} with respect to $V_m$.
\end{defi}

$\PP_m v$ is the piecewise linear interpolant of the values of $v$ on $V_m$. 

\begin{lemma}\label{P:Rayleigh-Ritz-1}
For any $m\ge 1$, let $V_m$ and $\PP_m$  be the Rayleigh-Ritz projection defined as in Definition~\ref{D:Rayleigh-Ritz}. Then for any $v\in\Dom\E$,
$\PP_m v$ is the component of $v$ in the subspace $S_D^m$, $v-\PP_m v$ vanishes on the boundary $\{a,b\}$, and
\begin{equation*}
    \E(v-\PP_m v, w)=0
    \quad\text{for all }w\in S_D^m.
\end{equation*}
\end{lemma}

\begin{proof} See, e.g., \cite{Strang-Fix_1973}.
\end{proof}

\begin{lemma} \label{th:projection} Assume the same hypotheses of Lemma~\ref{P:Rayleigh-Ritz-1}. Then for any $v\in\Dom\E$,
\begin{displaymath}
    v|_{V_m}=\PP_mv|_{V_m}.
\end{displaymath}
\end{lemma}

\begin{proof} Similar to that of \cite[Lemma 5.3]{Chen-Ngai_2010}.
\end{proof}

Let $\|V_m\|:=\max\{x_i-x_{i-1}:1\le i\le m\}$ denote the \textit{norm} of the partition $V_m$.

\begin{lemma}\label{th:convproj} Assume the same hypotheses of Lemma~\ref{P:Rayleigh-Ritz-1}
and let $v\in\Dom\E$.
Then
$$
\left|\PP_mv(x)-v(x)\right|\leq 2 \left\|V_m\right\|^{1/2}\|v\|_{\Dom\E}\quad\text{for all } x\in[a,b].
$$
In particular,
$$
\Vert \PP_mv-v\Vert_{\mu}\le 2\Vert V_m\Vert^{1/2}\Vert v\Vert_{\Dom\E}.
$$
\end{lemma}

\begin{proof} We first note that since $v$ is absolutely continuous and belongs to $\Dom\E$,
\begin{equation}\label{eq:abs_cont_Holder}
\left| v(x)-v(y) \right| =\left| \int_y^x v'(s)\,ds \right|\leq \left|x-y\right|^{1/2}\left\| v \right\|_{H_0^1(a,b)},\quad\forall x,y\in[a,b].
\end{equation}

Now Let $i\in \{1,\dots,N^m\}$ so that $x\in [x_{i-1},x_i]$. Then by \eqref{eq:abs_cont_Holder} and Lemma~\ref{th:projection}, we get
\begin{equation*}
\begin{aligned}
|\PP_m v(x)-v(x)|&\le|\PP_m v(x)-v(x_{i-1})|+|v(x_{i-1})-v(x)|\\
&\le|v(x_i)-v(x_{i-1})|+|v(x_{i-1})-v(x)|\\
&\le 2\|V_m\|^{1/2}\left\|v \right\|_{\Dom\E}.
\end{aligned}
\end{equation*}
\end{proof}

Throughout the rest of this section we let
\begin{equation}\label{E:ghf_conv_pf}
g,h\in\Dom\E\quad\text{and}\quad f=0,
\end{equation}
and let $u$ be the solution of the corresponding homogeneous IBVP~\eqref{eq:IBVP}. According to Theorem~\ref{T:smoothness},
\begin{equation}\label{E:u_conv_pf}
u\in W_2^k(0,T;\Dom\E)\quad\text{for all }k\ge 0.
\end{equation}
In particular, $u_{tt}\in\Dom\E$ and
\begin{equation}\label{E:weak_solution_conv_pf}
(u_{tt},v)_\mu+\E(u,v)=0\quad\text{for all }v\in\Dom\E.
\end{equation}
As in Section~\ref{S:FEM}, we let
\begin{equation*}
u^m(x,t)=\sum_{i=1}^{N^m-1}\beta_i(t)\phi_i(x).
\end{equation*}
Lastly, we define
\begin{equation*}
e(x,t)=e^m(x,t):=\PP_m u(x,t)-u^m(x,t).
\end{equation*}

\begin{lemma}\label{approximation-estimate}
Let $g,h,f,u,u^m,e$ be defined as above.

\begin{enumerate}
\item[(a)] $u^m$ satisfies:
\begin{enumerate}
\item[(i)] $\displaystyle(u^m_{tt},v^m)_{\mu}+\E(u^m,v^m)=0\text{ for all } v^m \in S_D^m,$
    \medskip
\item[(ii)] $u^m(x,0)=\sum_{i=1}^{N^m-1} g(x_i)\phi_i(x)
\text{ and }\ u^m_t(x,0)=\sum_{i=1}^{N^m-1} h(x_i)\phi_i(x).$
\end{enumerate}

\medskip
\item[(b)] The following identity holds:
\begin{equation}\label{eq:e-identity}
(e_{tt},e_t)_{\mu}+\E(e,e_t)=(\PP_mu_{tt}-u_{tt},e_t)_{\mu}.
\end{equation}
\end{enumerate}
\end{lemma}

\begin{proof} (a) The proof of part (a) follows from the derivations in Section~\ref{S:FEM}; we omit the details.

(b) By definition and the fact that $u\in W_2^k(0,T;\Dom\E)$ for $k\ge 0$, the functions $e_t,e_{tt}$, and $(\PP_m u)_{tt}=\PP_mu_{tt}$ all belong to $\SSS_D^m$.

Substitute $e_t$ for $v$ in \eqref{E:weak_solution_conv_pf} and for $v^m$ in (a)(i), and then subtracting the resulting equations, we get
$$
(u_{tt}-u^m_{tt},e_t)_\mu+\E(u-u^m,e_t)=0.
$$
 Equivalently,
$$
(u_{tt}-\PP_m u_{tt}+\PP_m u_{tt}-u^m_{tt},e_t)_\mu +\E(u-\PP_m u+\PP_m u-u^m,e_t)=0,
$$
which implies
$$
(\PP_m u_{tt}-u^m_{tt},e_t)_{\mu}+\E(\PP_m u-u^m,e_t)=(\PP_m u_{tt}-u_{tt},e_t)_\mu,
$$
because $\E(u-\PP_m u,e_t)=0$ (Lemma~\ref{P:Rayleigh-Ritz-1}).
Identity \eqref{eq:e-identity} now follows from the definition of $e(t)$.
\end{proof}

\begin{theo}\label{th:convproj2}
Assume the same hypotheses of Lemma~\ref{approximation-estimate} and let $\rho$ be as in \eqref{eq:rho_def}. Then there exists a constant $C>0$ such that
 $$
 \left\|\PP_m u-u^m\right\|_{\mu}\leq  C \sqrt{T} \rho^{m/2}\left\|u_{tt}\right\|_{2,\Dom\E}.
 $$
\end{theo}

\begin{proof}
Let $E(t):=\frac{1}{2}(e_t,e_t)_{\mu}+\frac{1}{2}\E(e,e)=
\frac{1}{2}\left\|e_t\right\|^2_{\mu}+\frac{1}{2}\left\|e\right\|^2_{\Dom{\E}}.$
Then
\begin{equation}\label{combine-form4}
\left\|e_t\right\|_{\mu} \leq \sqrt{2} \sqrt{E(t)},
\end{equation}
\begin{equation}\label{combine-form5}
\left\|e\right\|_{\Dom{\E}} \leq \sqrt{2} \sqrt{E(t)},
\end{equation}
\begin{equation}\label{combine-form6}
E(t)\leq \frac{1}{2}\big(\left\|e_t\right\|_{\mu}+\left\|e\right\|_{\Dom{\E}}\big)^2.
\end{equation}

The left-hand side of \eqref{eq:e-identity} is equal to
\begin{equation}\label{eq:left-side-combine-form3}
\frac{1}{2}\big(\left\|e_t\right\|^2_{\mu}\big)_t+\frac{1}{2}\big(\left\|e\right\|^2_{\Dom{\E}}\big)_t = E_t(t).
\end{equation}
For the right-hand side of \eqref{eq:e-identity}, we apply Cauchy-Schwarz inequality and \eqref{combine-form4} to get
\begin{equation}\label{eq:right-side-combine-form3}
E_t(t)=(\PP_m u_{tt}-u_{tt},e_t)_\mu\le\Vert\PP_m u_{tt}-u_{tt}\Vert_\mu\left\|e_t\right\|_{\mu}
\leq\|\PP_m u_{tt}-u_{tt}\|_{\mu}\sqrt{2}\sqrt{E(t)}.
\end{equation}

Since $E(t)\ge 0$ with $E(0)=0$, we can assume that $E(t)>0$ on some interval $(\alpha,\beta)\subset [0,T]$ with $\alpha<\beta$ and $E(\alpha)=0$. (Otherwise, by the continuity $E(t)$, we have $E(s)=0$ for all $s\in[0,T]$ and (\ref{energy-form02}) below still holds.) It follows from \eqref{eq:right-side-combine-form3} that
$$
 \frac{E_t(t)}{\sqrt{E(t)}} \leq \sqrt{2}\|\PP_m u_{tt}-u_{tt}\|_{\mu},\quad\alpha<s<\beta,
$$
and thus
\begin{equation}\label{energy-form02}
 2 {\sqrt{E(s)}}   \leq  \sqrt{2} \int_\alpha^\beta \|\PP_m u_{tt}-u_{tt}\|_{\mu} \,dt,\quad\alpha\le s\le\beta.
\end{equation}
From \eqref{combine-form5} and \eqref{energy-form02}, we have
\begin{equation*}
\left\|e(s)\right\|_{\Dom{\E}} \leq  \sqrt{2} \sqrt{E(s)} \leq \sqrt{2}  \int_\alpha^\beta \big\|\PP_m u_{tt}-u_{tt}\|_{\mu}\,dt\le \sqrt{2}\int_0^T \big\|\PP_m u_{tt}-u_{tt}\big\|_{\mu} \,dt,
\end{equation*}
which actually holds for all $s\in[0,T]$. 
Thus by combining condition \eqref{eq:condition_C}, Lemma~\ref{th:convproj}, and the above estimations, we have
\begin{equation*}\label{energy-form04}
\begin{split}
\left\|e(s)\right\|_{\mu}
\leq & C\left\|e(s)\right\|_{\Dom{\E}}
\leq  C \sqrt{T}\Big(\int_0^T \left\|\PP_m u_{tt}-u_{tt}\right\|^2_{\mu} \,dt\Big)^{1/2}\\
\leq &C \sqrt{T} \Big(\int_0^T \big(2\Vert V_m\Vert^{1/2}\left\|u_{tt}\right\|_{\Dom\E}\big)^2\,dt\Big)^{1/2}
\qquad\text{(Lemma~\ref{th:convproj})} \\
\leq &2 C \sqrt{T}\rho^{m/2}\left\|u_{tt}\right\|_{2,\Dom\E},
\end{split}
\end{equation*}
which holds for all $s\in[0,T]$. This completes the proof.
\end{proof}

\begin{proof}[\textit{Proof of Theorem~\ref{T:main-3}}] For fixed $t\in [0, T]$,
\begin{equation*}\label{energy-form05}
\begin{split}
\left\|u^m-u\right\|_{\mu} \leq &\left\|u^m-\PP_m u\right\|_{\mu} +\left\|\PP_m u-u\right\|_{\mu}.
\end{split}
\end{equation*}
Theorem~\ref{T:main-3} now follows by combining Lemma~\ref{th:convproj} and Theorem~\ref{th:convproj2}.
\end{proof}

{\bf Acknowledgements.}~~The authors thank Scott Kersey, Yin-Tat Lee, Frederic Mynard, Po-Lam Yung, Shijun Zheng, and, especially, Robert Strichartz, for valuable discussions and suggestions.

\end{document}